\documentclass[copyright,creativecommons]{eptcs}
\providecommand{\event}{QPL 2013} % Name of the event you are submitting to
\usepackage{breakurl}             % Not needed if you use pdflatex only.

\usepackage{enumerate,amsthm,amsmath,amssymb,color,graphicx,bbm}
\usepackage[square,sort,comma,numbers]{natbib} %\usepackage{natbib}
\usepackage[all]{xy}
\usepackage{tikz}
\usepackage{subfigure}

\usetikzlibrary{calc,decorations.pathreplacing}
\usetikzlibrary{arrows,shapes}

\usepackage{hyperref}
\newtheorem{theor}{Theorem}

\newcommand{\be}{\begin{equation}}
\newcommand{\ee}{\end{equation}}

\DeclareMathAlphabet{\mathpzc}{OT1}{pzc}{m}{it}

\newcommand{\NS}{\mathpzc{NS}}
\newcommand{\G}{\mathpzc{G}}
\newcommand{\Q}{\mathpzc{Q}}
\newcommand{\CE}{\mathpzc{CE}}
\newcommand{\C}{\mathpzc{C}}
\newcommand{\HS}{\mathpzc{H}}
\newcommand{\B}{\mathpzc{B}}
\newcommand{\LO}{\mathpzc{LO}}

\title{Probabilistic models on contextuality scenarios}
\author{Tobias Fritz
\institute{Perimeter Institute for Theoretical Physics, Waterloo, Ontario, Canada}
\email{tfritz@perimeterinstitute.ca}
\and
Anthony Leverrier
\institute{INRIA Rocquencourt, Domaine de Voluceau, B.P. 105, 78153 Le Chesnay Cedex, France}
\email{anthony.leverrier@inria.fr}
\and
Ana Bel{\'e}n Sainz
\institute{ICFO--Institut de Ciencies Fotoniques, E--08860 Castelldefels, Barcelona, Spain}\email{belen.sainz@icfo.es}
}
\def\titlerunning{Probabilistic models on contextuality scenarios}
\def\authorrunning{T. Fritz, A. Leverrier \& A.B. Sainz}
\begin{document}
\maketitle

\begin{abstract}
We introduce a framework to describe probabilistic models in Bell experiments, and more generally in \emph{contextuality scenarios}. Such a scenario is a hypergraph whose vertices represent elementary events and hyperedges correspond to measurements. A probabilistic model on such a scenario associates to each event a probability, in such a way that events in a given measurement have a total probability equal to one. We discuss the advantages of this framework, like the unification of the notions of contexuality and nonlocality, and give a short overview of the results obtained in Ref.~\cite{FLS12}.
\end{abstract}

The main goal of physics is to understand how Nature works, and usually, physicists proceed as follows: first, observe a phenomenon, then propose a model that explains it, extract predictions from this model, and, finally, confront these predictions with experimental data. Repeat until the experimental results match the theoretical predictions. In some situations, however, it can be fruitful to limit the model to a minimum. This idea was recently investigated in the paradigm of \emph{device-independence} \cite{BCP13}. 
There, an experimenter has access to a physical device with classical commands $x \in \mathcal{X}$ and classical results $a \in \mathcal{A}$ and chooses not to model the inner workings of the device any further. This might seem futile at first sight: how can one hope to say anything meaningful when only observing conditional probabilities of the form $P(a|x)$, corresponding to the probability of obtaining outcome $a$ when applying command (or measurement) $x$?
The key idea is to consider $n$ physical devices used in a space-like separated way by $n$ experimenters. Then, one has access to the conditional probability distribution $P(a_1 \ldots a_n|x_1 \ldots x_n)$ with $a_i$ and $x_i$ referring to the outcomes and measurements of the $i^{\mathrm{th}}$ party, where the no-signaling principle constrains $P$ non-trivially. Stronger restrictions can be imposed by requiring the devices to be compatible with quantum theory, or even to be classical. 
In this paper, we summarize a framework allowing to describe such \emph{Bell-type} scenarios in a very general way, and that extends naturally to contextuality scenarios. See~\cite{FLS12} for more details.

\section{Contextuality scenarios}
 We define a \emph{contextuality scenario} to be a hypergraph $H = (V,E)$ whose vertices $v \in V$ correspond to the events of the scenario, and the hyperedges $e = \{v_1, \cdots, v_k\} \in E$ are subsets of $V$ that should be thought of as the measurements of the scenario. We demand in addition that all the vertices belong to at least one hyperedge. 
 Such scenarios have been studied before in quantum logic where they are known as ``test spaces'' \cite{Wil09}.
A \emph{probabilistic model} on the scenario $H$ is then given by an assignment $p: V \rightarrow [0,1]$ of a probability $p(v)$ to each event $v \in V$ satisfying the normalization condition $\sum_{v \in e} p(v) =1$ for each measurement $e \in E$. Let us denote by $\G(H) \subseteq [0,1]^{|V|}$ the set of probabilistic models for the scenario $H$. By construction, this set is a polytope, the set of ``states'' on $H$ in the terminology of test spaces.
Let us note that this approach was inspired by the framework developed in \cite{CSW10}, but that a crucial difference between the two works is that we explicitly work with normalized probability distributions, instead of subnormalized ones.
 
\section{Bell-type scenarios}
 An important application of this framework concerns Bell-type scenarios where $n$ parties have access to $n$ distinct devices. For simplicity, we restrict ourselves to the scenario $\B_{n,m,k}$, where the $n$ devices all have $m$ different settings and $k$ possible outcomes. In particular, $\B_{2,2,2}$ will correspond to the usual CHSH scenario. We now describe the hypergraph $\B_{n,m,k}$. Its vertices are the $(mk)^n$ events of the form $(a_1 \ldots a_n | x_1\ldots x_n)$. The trickier part is to characterize the measurements of the scenario. Usually, one would define a measurement to be the set of events of the form $(\cdot |x_1 \ldots x_n)$ for fixed settings $x_i$. However, our framework includes additional measurements: a measurement in the scenario $\B_{n,m,k}$ corresponds to any strategy applied by the $n$ parties, \emph{possibly coming together}, where each of the parties measures their device. More specifically, a measurement of $\B_{n,m,k}$ is given by a temporal ordering of the parties: $i_1 \leq i_2 \leq \ldots \leq i_n$ where party $i_1$ first chooses a measurement setting $x_{i_1}$ and obtains an outcome $a_{i_1}$. Then, party $i_2$ chooses a setting $x_{i_2}$, possibly depending on $x_{i_1}$ and $a_{i_1}$, and obtains an outcome $a_{i_2}$. This process is repeated until the last party performs their measurement. Note that the strategy can be adaptive, meaning that party $i_k$ can choose their measurement setting to be a function of the previous outcomes $x_{i_1},\ldots, x_{i_{k-1}}$. In fact, in the most general kind of measurement allowed by our definition of Bell scenario, even the order of the parties may be adaptive in the sense that it may depend on previous outcomes. The scenario obtained this way is displayed on Fig.~\ref{CHSH} for the case of $\B_{2,2,2}$. Similarly general measurements have also been considered in~\cite{SB}.

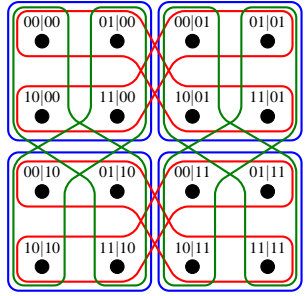
\begin{figure}
\definecolor{darkgreen}{rgb}{0,.5,0}
\begin{center}
\begin{tikzpicture}
\clip (0,.5) rectangle (5,5);
\foreach \x in {0,1} \foreach \y in {0,1} \foreach \a in {0,1} \foreach \b in {0,1}
{
	\node[draw,shape=circle,fill,scale=.5] (e) at (\b+2*\y+1,4-\a-2*\x) {} ;
	\node[above=0pt] at (e) {\tiny{$\a\b|\x\y$}} ;
}
\foreach \x in {0,2} \foreach \y in {0,2} \draw[rounded corners,thick,blue] (\x+0.55,\y+0.68) rectangle (\x+2.45,\y+2.52) ;
\foreach \y in {0,2} \draw[thick,red,rounded corners] (2.8,\y+0.8) -- (4.33,\y+0.8) -- (4.33,\y+1.4) -- (2.8,\y+1.4) -- (2.25,\y+2.4) -- (0.67,\y+2.4) -- (0.67,\y+1.8) -- (2.2,\y+1.8) -- cycle ;
\foreach \y in {0,2} \draw[thick,red,rounded corners] (2.8,\y+1.8) -- (4.33,\y+1.8) -- (4.33,\y+2.4) -- (2.75,\y+2.4) -- (2.2,\y+1.4) -- (0.67,\y+1.4) -- (0.67,\y+0.8) -- (2.2,\y+0.8) -- cycle ;
\foreach \x in {0,2} \draw[thick,darkgreen,rounded corners] (\x+0.62,2.4) -- (\x+0.62,0.75) -- (\x+1.3,0.75) -- (\x+1.3,2.1) -- (\x+2.38,2.8) -- (\x+2.38,4.45) -- (\x+1.65,4.45) -- (\x+1.65,3) -- cycle ;
\foreach \x in {0,2} \draw[thick,darkgreen,rounded corners] (5-\x-0.62,2.4) -- (5-\x-0.62,0.75) -- (5-\x-1.3,0.75) -- (5-\x-1.3,2.1) -- (5-\x-2.38,2.8) -- (5-\x-2.38,4.45) -- (5-\x-1.65,4.45) -- (5-\x-1.65,3) -- cycle ;
\end{tikzpicture}
\end{center}
\caption{The 16 events and 12 measurements of the CHSH scenario, $\B_{2,2,2}$}
\label{CHSH}
\end{figure}

The main advantage of defining $\B_{n,m,k}$ as above is that $\G(\B_{n,m,k})$ is exactly the standard no-signaling polytope $\NS(\B_{n,m,k})$, defined as correlations satisfying 
$$\sum_{a_{i+1} \ldots a_n} p(a_1 \ldots a_n | x_1 \ldots x_n) = p(a_1 \ldots a_i | x_1\ldots x_i)$$ 
for any splitting of the $n$ parties into two groups. This may seem surprising, since some hyperedges of $\B_{n,m,k}$ correspond to correlated measurements among the parties, where they communicate to each other. However, these measurements are exactly the ones which guarantee the no-signaling properties of the allowed probabilistic models. The proof that $\G(\B_{n,m,k})=\NS(\B_{n,m,k})$ is straightforward~\cite{FLS12} and here we only give the intuition in the case of $\B_{2,2,2}$. We wish to show that the normalization of the hyperedges (i.e.~that the total probability of the events in any measurement is 1) is equivalent to the no-signaling condition. A typical no-signaling condition for CHSH reads: $p(00|00)+ p(01|00) = p(00|01)+p(01|01)$ (corresponding to the first row on Fig.~\ref{CHSH}). This can be derived from the normalization of the measurement ``00'' consisting of events of the form $(\cdot|00)$ and implying that $p(00|00)+ p(01|00)=1-p(10|00)- p(11|00)$ and of the event $\{(10|00), (11|00), (00|01), (01|01)\}$ implying that $p(00|01)+p(01|01)$ is also equal to $1-p(10|00)- p(11|00)$. Hence, normalization implies no-signaling and the converse property can also be checked in the same fashion. 

\section{Classical and quantum models}
There are two natural restrictions that one might want to impose on the devices: either of a classical, or a quantum nature, leading respectively to the notions of \emph{classical} and \emph{quantum} probabilistic models. 
First, a \emph{deterministic} model on $H$ is a probabilistic model (hence satisfying normalization) such that $p(v) \in \{0, 1\}$ for all events $v\in V$. Then, classical models are given by convex combinations of deterministic models: $p(v) = \sum_\lambda q_\lambda p_\lambda(v)$, where $q_\lambda$ is a probability distribution, and every $p_\lambda$ is a deterministic model on $H$. The set of classical models on $H$ is denoted by $\C(H)$. If $H$ is a Bell-type scenario, then $\C(H)$ is the standard Bell polytope. If $H$ is a general contextuality scenario, classical models are those that can be explained by noncontextual hidden variables \cite{Fin82}. 

A quantum model $p$ on $H$ is a probabilistic model such that there exist a Hilbert space $\HS$, a normalized density matrix $\rho \in \B(\HS)$, and for each vertex $v \in V$, a projector $P_v$ such that $\sum_{v \in e} P_v = \mathbbm{1}_\HS$ for each measurement $e\in E$ that give rise to $p$ via the Born rule: $p(v) = \mathrm{tr} \, (\rho P_v)$, for each event $v$. The set of quantum models on $H$ is denoted by $\Q(H)$. Contrary to $\C(H)$ and $\G(H)$, the quantum set is usually not a polytope, and a recurring question in the literature is to find some ``natural principle'' that limits correlations observable in Nature to be those in the quantum set. Since $\Q(\B_{2,2,2}) \subsetneq \NS(\B_{2,2,2})$, it is clear that the no-signaling principle alone is not sufficient to restrict the correlations to be quantum. 

\section{The quantum set from a natural principle}
Several such candidate principles have been suggested and investigated: Information Causality \cite{PPK09}, Macroscopic Locality \cite{NW10}, the nontriviality of communication complexity \cite{vD05}, and more recently, Local Orthogonality \cite{FSA12}. The latter is particularly interesting in the sense that it is a genuinely multipartite principle, a necessary condition in order to recover the quantum set \cite{GWA11}. The framework we introduced above turns out to be remarkably well-suited for the study of Local Orthogonality (LO). The principle defines a notion of orthogonality between events of a Bell scenario, which in the language of this work is expressed as follows: two events $u$ and $v$ are orthogonal if they belong to a common measurement, i.e., there exists a measurement $e \in E$ such that $\{u, v\} \subseteq e$. Then, a set $C = \{v_1, \cdots, v_l\} \subseteq V$ of events is said to be \emph{orthogonal} if its elements are pairwise orthogonal. 
The principle finally says that the sum of the individual probabilities of a set of orthogonal events is at most one, $\sum_{v \in C} p(v) \leq 1$. 
The set obtained this way is a polytope denoted by $\LO^1(H)$. 
In our framework, the LO principle turns to be equivalent to the \emph{Consistent Exclusivity} principle for general contextuality scenarios \cite{Henson, cabello}, hence we will focus on the study of the latter. 
A natural strengthening of the CE principle assumes that if a given probabilistic model is ``physical'', then the same should apply to an arbitrary number $k$ of copies of this model. Then, Consistent Exclusivity should also be satisfied by the model corresponding to these $k$ copies. Copies of a scenario can be defined via the $k$-fold \emph{Foulis-Randall} product of the scenario $H$ with itself, $H^{\otimes k}$. The Foulis-Randall product \cite{FR81} is especially relevant in the context of Bell scenarios: scenarios with many parties can be obtained by taking the product of several single-party scenarios. In particular, $\B_{n,m,k} = \B_{1,m,k}^{\otimes n}$. Now, the strengthening of CE says that the resulting product probabilistic model $p^{\otimes k} \in \G(H^{\otimes k})$ should also satisfy CE. We denote by $\CE^k(H)$ the set of probabilistic models on $H$ such that $p^{\otimes k} \in \CE^1(H^{\otimes k})$. Note that for Bell scenarios $\CE^k(\B_{n,m,k}) = \LO^k(\B_{n,m,k})$, where the latter set was defined in \cite{FSA12}. In the limit of an arbitrary number of copies, this gives rise to the set $\CE^\infty(H)$, which would ideally match the set $\Q(H)$, were the CE principle sufficient to recover quantum correlations.
We note that another way to naturally strengthen CE would be to allow for wirings of boxes. However, it has been proved in \cite{FSA12} that these leave the set $\CE^\infty(H)$ invariant. 
It turns out that characterizing the set $\CE^\infty(H)$ of correlations satisfying the CE principle is quite challenging. While it is reasonably easy to verify that $\Q(H) \subseteq \CE^\infty(H) \subseteq\G(H)$, saying much more is difficult. 

Our framework, however, allows for a reformulation of $\CE^\infty(H)$ in terms of graph invariants. 
Introduce the \emph{non-orthogonality graph} $G = \mathrm{NO}(H)$ of the contextuality scenario $H$ to be the undirected graph with vertex set $V(H)$, and such that $\{u, v\}$ is an edge if $u$ and $v$ do not belong to a common measurement $e \in E(H)$.
Then, one can show~\cite{FLS12} that a probabilistic model $p$ belongs to $\CE^\infty(H)$ if and only if $\Theta(\mathrm{NO}(H),p) =1$ where $\Theta(G,p)$ refers to the Shannon capacity of the graph $G$ weighted by the distribution $p$.
This characterization can then be used to prove that $\CE^\infty(H)$ is in general strictly larger than $\Q(H)$~\footnote{In fact, a proof that the sets $\CE^\infty(H)$  and $\Q(H)$ are not equal was found by Miguel Navascu\'es before this formalism had been set up.}, and that there even exist contextuality scenarios for which $\CE^\infty(H)$ is not convex \cite{FLS12}.

\section{Hierarchies}
Another feature of our framework is that the various sets of correlations we mentioned can be approximated through some hierarchies of relaxations. Such hierarchies have been intensely studied in convex optimization (see Ref.~\cite{Lau09} for a recent review) and have been extended to noncommutative polynomial optimization~\cite{DLT08,PNA10}, including a characterization of quantum correlations in Bell scenarios~\cite{NPA07}. 

Let us first introduce the notion of moment matrix associated with a contextuality scenario $H=(V,E)$. A moment matrix of order $k$ associated with $H$ is a symmetric matrix $M_k$ whose rows and columns are indexed by \emph{words} of size at most $k$ written in the alphabet formed by $V$. More explicitly, if $V = \{v_1, \ldots, v_n\}$, the rows of the moment matrix will be indexed by: 
$$\emptyset, v_1, \ldots, v_n, v_1v_1, v_1v_2, \ldots, v_1v_n, \ldots, v_n v_n, v_1^3, \ldots, v_n^k,$$ where $v_i^k$ is the word obtained by concatenating $k$ times the letter $v_i$. Here, $\emptyset$ refers to the empty string, and we choose the normalization $M_k(\emptyset, \emptyset) = 1$. We denote by $V^*$ the set of strings of arbitrary size on $V$.
A matrix $M_k$ will be a \emph{certificate of order} $k$ for the probabilistic model $p$ on $H$ if it is positive semidefinite, $M_k \succeq 0$, and if $M_k(v,\emptyset)  = p(v)$ for every $v \in V$.

The matrices $M_k$ can display additional ``natural'' properties that we define now: Normalization, Orthogonality and Commutativity. A moment matrix is \emph{normalized} with respect to the contextuality scenario $H=(V,E)$ if for every two strings $\vec{v}, \vec{w} \in V^*$, and every hyperedge $e \in E$, the following condition holds:
\begin{align}
\sum_{u\in e} M(\vec{v} u, \vec{w}) = M(\vec{v}, \vec{w}).  \tag{Normalization}
\end{align}
A matrix is \emph{orthogonal} with respect to $H$ if for every $e\in E$, and $\vec{v}, \vec{w} \in V^*$, the fact that $v, w \in e$ $(v \ne w)$ implies that
\begin{align}
M(\vec{v} v, \vec{w} w) = 0 \quad \forall \vec{v}, \vec{w} \in V^*.  \tag{Orthogonality}
\end{align}
Finally, a matrix is \emph{commutative} if for any two strings $\vec{v}, \vec{w} \in V^*$, and every permutation $\pi$ of size $|\vec{v}|$, 
\begin{align}
M(\pi(\vec{v}),\vec{w}) = M(\vec{v}, \vec{w}),  \tag{Commutativity}
\end{align}
where $\pi(\vec{v})$ is the string obtained by permuting the letters of $\vec{v}$ with the permutation $\pi$.

We are now in a position to define sets of models for which there exist certificates satisfying some of these properties. These sets actually form hierarchies of sets $\left(\mathcal{S}_k\right)_{k\geq 1}$, such that $\mathcal{S}_{k} \subseteq \mathcal{S}_{k-1}$ is the set of probabilistic models with a certificate of order $k$. The hierarchies we will introduce admit limits that we denote by $\mathcal{S}_\infty := \bigcap_{k\geq 0} \mathcal{S}_k$. 
Let us define three hierarchies of sets $\G_k, \Q_k$ and $\C_k$ as follows. 
A probabilistic model $p$ on $H$ belongs to $\G_k(H)$ if there exists a certificate of order $k$ for $p$ satisfying Normalization; it belongs to $\Q_k(H)$, if there exists a certificate of order $k$ satisfying Normalization and Orthogonality; and it belongs to $\C_k(H)$ if there exists a certificate of order $k$ satisfying Normalization, Orthogonality and Commutativity. 
Our results show that these hierarchies converge to the expected sets.
\begin{theor}[Convergence of the hierarchies] \nonumber
For every contextuality scenario $H=(V,E)$, 
\begin{align}
\G_\infty(H) &= \G_1(H) = \G(H),\\
\Q_\infty (H)&= \Q(H),\\
\C_\infty (H)&= \C_{|V|}(H) = \C(H).
\end{align}
\end{theor}
\begin{proof}
The fact that $\G_1(H)=\G(H)$ holds by definition. Moreover, if $p\in \G(H)$, then one can construct an explicit certificate of any order by fixing: $M(v_1 \ldots v_n, w_1 \ldots w_m) := \prod_{i=1}^n p(v_i) \prod_{j=1}^m p(w_j)$, which is of rank 1 and clearly satisfies Normalization.

Given a quantum model $p \in \Q(H)$, together with its associated Hilbert space $\HS$, density matrix $\rho \in \B(\HS)$, and projectors $P_v$ for each $v\in V$, one can define $M(v_1 \ldots v_n, w_1 \ldots w_m) := \mathrm{tr}\,\left(\rho \prod_{i=1}^n P_{v_i} \prod_{j=m}^1 P_{w_j}\right)$. It is straightforward to check that this (infinite) matrix is positive semidefinite and satisfies both Normalization and Orthogonality. Alternatively, one needs to show that if such a certificate of order $k$ can be associated with $p$ for any $k \geq 0$, then it it possible to find a quantum model for $p$. This is done via the Gelfand-Naimark-Segal (GNS) construction by interpreting the infinite matrix $M$ as a $*$-algebraic state through the assignment $\phi(P_{v_1} \ldots P_{v_n}) = M(v_1 \ldots v_n, \emptyset)$ on the $*$-algebra with generators $\{P_v, v \in V\}$,  and relations $P_v = P_v^2 = P_v^*$ and $\sum_{v\in e} P_v = \mathbbm{1}$ for all $e\in E$. The GNS constructions then turns it into a quantum model satisfying $p(v) = \phi(P_v)$ for all $v\in V$. Full details of the proof are presented in Ref.~\cite{FLS12}.

Consider finally a model $p \in \C_\infty(H)$. By definition, if it is not empty, $\C_\infty(H)$ is contained in $\Q_\infty(H) = \Q(H)$. Because of the commutativity property and the fact that repeating a letter does not change the value of the entry (itself a consequence of Normalization and Orthogonality), it is clear that the sequence $\left(\C_k(H)\right)_{k\geq 1}$ converges after at most $|V|$ steps (since no ``new'' word can be formed with more letters). The projectors $P_v$ obtained from the GNS construction commute and can all be diagonalized in the same orthonormal basis $\Lambda = \left\{ |\lambda\rangle \right\}$. Expressing the associated density matrix $\rho \in \B(\HS)$ in the same basis, and denoting by $\tilde{\rho}$ the diagonal density matrix with the same diagonal as $\rho$, it is clear that $\tilde{\rho}$ gives rise to the same model as $\rho$. Writing $\tilde{\rho} = \sum_{\lambda \in \Lambda} q_\lambda |\lambda\rangle \langle \lambda|$, one obtains that for all $v\in V$, $p(v) = \sum_{\lambda \in \Lambda} q_\lambda \langle \lambda |P_v|\lambda\rangle$, where the distribution $\left(\langle \lambda|P_v| \lambda \rangle\right)_{\lambda \in \Lambda}$ corresponds to a deterministic model on $H$. Hence, $p(v) \in \C(H)$.

Conversely, for any classical model $p \in \C(H)$, there exist a probability distribution $(q_\lambda)_{\lambda \in \Lambda}$ and deterministic models $p_\lambda$ on $H$ for each $\lambda \in \Lambda$ such that $p(v) = \sum_{\lambda \in \Lambda} q_\lambda p_\lambda(v)$. Define a Hilbert space with basis $\{|\lambda\rangle \: : \: \lambda \in \Lambda\}$, projectors $P_v = \sum_{\lambda \in \Lambda} p_\lambda(v) |\lambda\rangle \langle \lambda|$ for all $v\in V$, and the diagonal density matrix $\rho = \mathrm{diag} (q_{\lambda_1}, q_{\lambda_2}, \ldots)$. It is straightforward to check that the matrix $M$ defined by $M(v_1 \ldots v_n, w_1 \ldots w_m) := \mathrm{tr}\,\left(\rho \prod_{i=1}^n P_{v_i} \prod_{j=1}^m P_{w_j}\right)$ is a certificate of any order satisfying Normalization, Orthogonality and Commutativity.
\end{proof}

The hierarchies $(\G_k)_{k\geq 1}$ and $(\C_k)_{k\geq 1}$ both converge after a finite number of steps, and it is natural to ask whether the same holds for $(\Q_k)_{k\geq 1}$. It is in fact an open question related to difficult problems in the theory of $C^*$-algebras whether there exist contextuality scenarios $H$ for which the hierarchy needs infinitely many steps to converge (see Section 8.3 of \cite{FLS12} for details).

\section{Link between $\CE^\infty(H)$ and the quantum set}
In the same way as $\CE^\infty(H)$ can be characterized via the Shannon capacity of the non-orthogonality graph $\mathrm{NO}(H)$, weighted by the distribution $p$, the first level of the quantum hierarchy, $\Q_1(H)$, can be characterized by the Lov\'asz number $\vartheta$ of $\mathrm{NO}(H)$, weighted by $p$. More precisely, a probabilistic model $p$ on the contextuality scenario $H$, belongs to $\Q_1(H)$ if and only if $\vartheta(\mathrm{NO}(H),p) =1$.

For every graph $G$, and any choice of weight $p$ for the vertices of $G$, it is known that $\Theta(G,p) \leq \vartheta(G,p)$, which immediately implies that for every contextuality scenario, $\Q_1(H) \subseteq \CE^\infty(H)$. This proves that the Local Orthogonality principle is not sufficient to recover the set of quantum correlations for arbitrary contextuality scenarios, since in general $\Q(H)\subsetneq\Q_1(H)$.

A possible strengthening of the CE principle is inspired by a recent paper \cite{Yan13}. One may take it as part of a principle to assume for granted that quantum correlations are physical \emph{for any contextuality scenario $H$} and then only look for a postulate that excludes the existence of probabilistic models outside $\Q(H)$. Using this idea, it is possible to show that the set of probabilistic models satisfying this extension of CE is no longer $\CE^\infty(H)$, but rather $\Q_1(H)$ (which is not equal to $\Q(H)$ in general). However, asking that a property like the existence of quantum models holds for any contextuality scenario may not be as natural as asking that it holds for Bell-type scenarios only.

To summarize, we have introduced a new framework for contextuality and nonlocality, that allows to treat Bell scenarios as a particular case of contextuality scenarios. This approach significantly refines that of~\cite{CSW10}, since it includes the normalization of the probabilistic models; only this allows us to recover Bell scenarios as special cases. Moreover, the description of Bell scenarios is instrinsically related to the existence of correlated measurements among the parties, and these are naturally described in terms of the Foulis-Randall product of contextuality scenarios. This framework is well-suited for studing correlations based on orthogonal events, such as those characterized from the Consistent Exclusivity principle or the Local Orthogonality principle. In particular, we defined the non-orthogonality graph of a contextuality scenario, and used it to characterize whether a probabilistic model belongs to $\CE^\infty(H)$ in terms of a graph invariant. We further defined a hierarchy of relaxations that converge to the quantum set, and used it to prove that in general $\Q(H) \subsetneq \CE^\infty(H)$. We believe that there exist other connections between this framework and other formalisms, that may be of great use for understanding the set of quantum models.

\section*{Acknowledgments}
We thank Antonio Ac{\'i}n and Miguel Navascu\'es for comments and discussion.
Research at Perimeter Institute is supported by the Government of Canada
through Industry Canada and by the Province of Ontario through the Ministry
of Economic Development and Innovation. This work has been supported by a grant from the John Templeton Foundation.
A.B.S. was supported by the ERC SG PERCENT and by the Spanish projects FIS2010-14830, DIQIP and spanish FPU:AP2009-1174 PhD grant.

\nocite{*}
\bibliographystyle{eptcs}
\bibliography{FLS}
\end{document}